\documentclass[aps,nopacs,nokeys,superscriptaddress,11pt,twoside,notitlepage]{revtex4-1}

\usepackage{graphicx,epic,eepic,epsfig,amsmath,latexsym,amssymb,amscd,verbatim,color}

\usepackage{theorem}

\newtheorem{definition}{Definition}
\newtheorem{proposition}[definition]{Proposition}

\newtheorem{theorem}[definition]{Theorem}

\newtheorem{observation}[definition]{Observation}

\def\squareforqed{\hbox{\rlap{$\sqcap$}$\sqcup$}}
\def\qed{\ifmmode\squareforqed\else{\unskip\nobreak\hfil
\penalty50\hskip1em\null\nobreak\hfil\squareforqed
\parfillskip=0pt\finalhyphendemerits=0\endgraf}\fi}
\def\endenv{\ifmmode\;\else{\unskip\nobreak\hfil
\penalty50\hskip1em\null\nobreak\hfil\;
\parfillskip=0pt\finalhyphendemerits=0\endgraf}\fi}
\newenvironment{proof}{\noindent \textbf{{Proof~} }}{\qed}

\mathchardef\ordinarycolon\mathcode`\:
\mathcode`\:=\string"8000
\def\vcentcolon{\mathrel{\mathop\ordinarycolon}}
\begingroup \catcode`\:=\active
  \lowercase{\endgroup
  \let :\vcentcolon
  }

\newcommand{\nc}{\newcommand}
\nc{\rnc}{\renewcommand}
\nc{\beq}{\begin{equation}}
\nc{\eeq}{{\end{equation}}}
\nc{\beqa}{\begin{eqnarray}}
\nc{\eeqa}{\end{eqnarray}}
\nc{\lbar}[1]{\overline{#1}}
\nc{\bra}[1]{\langle#1|}
\nc{\ket}[1]{|#1\rangle}
\nc{\ketbra}[2]{|#1\rangle\!\langle#2|}
\nc{\braket}[2]{\langle#1|#2\rangle}
\nc{\proj}[1]{| #1\rangle\!\langle #1 |}
\nc{\avg}[1]{\langle#1\rangle}
\nc{\Rank}{\operatorname{Rank}}
\nc{\smfrac}[2]{\mbox{$\frac{#1}{#2}$}}
\nc{\tr}{\operatorname{Tr}}
\nc{\ox}{\otimes}
\nc{\dg}{\dagger}
\nc{\dn}{\downarrow}
\nc{\cA}{{\cal A}}
\nc{\cB}{{\cal B}}
\nc{\cC}{{\cal C}}
\nc{\cD}{{\cal D}}
\nc{\cE}{{\cal E}}
\nc{\cF}{{\cal F}}
\nc{\cG}{{\cal G}}
\nc{\cH}{{\cal H}}
\nc{\cI}{{\cal I}}
\nc{\cJ}{{\cal J}}
\nc{\cK}{{\cal K}}
\nc{\cL}{{\cal L}}
\nc{\cM}{{\cal M}}
\nc{\cN}{{\cal N}}
\nc{\cO}{{\cal O}}
\nc{\cP}{{\cal P}}
\nc{\cR}{{\cal R}}
\nc{\cS}{{\cal S}}
\nc{\cT}{{\cal T}}
\nc{\cX}{{\cal X}}
\nc{\cZ}{{\cal Z}}
\nc{\csupp}{{\operatorname{csupp}}}
\nc{\qsupp}{{\operatorname{qsupp}}}
\nc{\var}{{\operatorname{var}}}
\nc{\rar}{\rightarrow}
\nc{\lrar}{\longrightarrow}
\nc{\polylog}{{\operatorname{polylog}}}
\nc{\1}{{\openone}}
\nc{\wt}{{\operatorname{wt}}}
\nc{\av}[1]{{\left\langle {#1} \right\rangle}}

\nc{\RR}{{{\mathbb R}}}
\nc{\CC}{{{\mathbb C}}}
\nc{\FF}{{{\mathbb F}}}
\nc{\NN}{{{\mathbb N}}}
\nc{\ZZ}{{{\mathbb Z}}}
\nc{\PP}{{{\mathbb P}}}
\nc{\QQ}{{{\mathbb Q}}}
\nc{\UU}{{{\mathbb U}}}
\nc{\EE}{{{\mathbb E}}}
\nc{\id}{{\operatorname{id}}}

\nc{\CHSH}{{\operatorname{CHSH}}}

\nc{\be}{\begin{equation}}
\nc{\ee}{{\end{equation}}}
\nc{\bea}{\begin{eqnarray}}
\nc{\eea}{\end{eqnarray}}
\nc{\<}{\langle}
\rnc{\>}{\rangle}
\nc{\Hom}[2]{\mbox{Hom}(\CC^{#1},\CC^{#2})}
\nc{\rU}{\mbox{U}}

\nc{\ob}[1]{#1}

\nc{\SEP}{{\text{SEP}}}
\nc{\sep}{{\text{sep}}}
\nc{\LOCC}{{\text{LOCC}}}
\nc{\PPT}{{\text{PPT}}}
\nc{\EXT}{{\text{EXT}}}
\nc{\ALL}{{\text{ALL}}}
\nc{\Sym}{{\operatorname{Sym}}}

\nc{\vertleq}{{\rotatebox[origin=c]{90}{$\leq$}}}

\begin{document}

\title{Distinguishing multi-partite states by local measurements}
\author{C\'{e}cilia Lancien}
\affiliation{Centre de Math\'{e}matiques Laurent Schwartz, \'{E}cole Polytechnique, 91128 Palaiseau Cedex, France}
\email{cecilia.lancien@polytechnique.edu}
\affiliation{Department of Mathematics, University of Bristol, Bristol BS8 1TW, U.K.}

\author{Andreas Winter}
\affiliation{ICREA \&{} F\'{\i}sica Te\`{o}rica: Informaci\'{o} i Fenomens Qu\`{a}ntics, Universitat Aut\`{o}noma de Barcelona, ES-08193 Bellaterra (Barcelona), Spain}
\email{der.winter@gmail.com}
\affiliation{Department of Mathematics, University of Bristol, Bristol BS8 1TW, U.K.}
\affiliation{Centre for Quantum Technologies, National University of Singapore, 2 Science Drive 3, Singapore 117542}

\date{6 December 2012}

\begin{abstract}
  We analyze the distinguishability norm on the states of a multi-partite
  system, defined by local measurements.
  Concretely, we show that the norm associated to
  a tensor product of sufficiently symmetric measurements is essentially
  equivalent to a multi-partite generalisation of the non-commutative
  $\ell_2$-norm (aka Hilbert-Schmidt norm): in comparing the two, the constants
  of domination depend only on the number of parties but not on the Hilbert
  spaces dimensions.

  We discuss implications of this result on the corresponding norms for
  the class of all measurements implementable by local operations and
  classical communication (LOCC), and in particular on the leading order
  optimality of multi-party data hiding schemes.
\end{abstract}

\maketitle


\section{Distinguishability norms}
\label{sec:intro}
The task of distinguishing quantum states from accessible experimental data
is at the heart of quantum information theory, appearing right at its
historical beginnings -- see \cite{Holevo,Helstrom}, and \cite{NielsenChuang}
for general reference.
Indeed, the special case on which we are focussing in this paper, the discrimination
of two states, is the generalisation of hypothesis testing in classical
statistics. There, the optimal discrimination between two hypotheses,
modelled as (for simplicity: discrete) probability distributions
$P_0$ and $P_1$, with prior probabilities $q$ and $1-q$, respectively, is
given by the maximum likelihood rule \cite{statistics}.
The minimum error probability is thus given by
\[
  \Pr\{\text{error}\} = \frac12\bigl( 1 - \| qP_0-(1-q)P_1 \|_1 \bigr),
\]
with the usual $\ell_1$-norm $\displaystyle{\| \Delta \|_1 = \sum_{x\in\cX} |\Delta_x|}$.

In this paper, we shall denote by the same symbol its non-commutative generalisation
$\| \Delta \|_1 = \tr |\Delta|$, i.e. the sum of the singular values of $\Delta$,
also known as trace norm.

Owing to the particular role played by \emph{measurement} in quantum
mechanics, however, any restriction on the set of available measurements
leads to a specific norm on density operators: any decision
in the discrimination task must be based on measurement results.
Specifically, let the two hypotheses be two quantum states (density operators)
$\rho_0$ and $\rho_1$ on some Hilbert space $\cH$,
with prior probabilities $q$ and $1-q$, respectively.
A generic measurement $M$, i.e. a positive operator valued measure (POVM,
aka partition of unity), is given by positive semidefinite operators
\[
  M_x \geq 0, \quad\text{s.t.}\quad \sum_{x\in\cX} M_x=\1.
\]
(In this paper, POVMs will generally be discrete and Hilbert spaces will always be of
finite dimension. With suitable adaptations to the proofs, however, our results
carry over to general POVMs and infinite dimension.)
The Born rule for measurements postulates that the state $\rho_i$
generates a distribution $P_i$ on the outputs of the measurement:
\[
  P_i(x) = \tr \rho_i M_x,
\]
and hence the minimum error probability in any decision based on $i$ is
\[\begin{split}
  \Pr\{\text{error}\} &=  \frac12\bigl( 1 - \| qP_0-(1-q)P_1 \|_1 \bigr) \\
                      &=  \frac12\left( 1 - \sum_{x\in\cX} \bigl|\tr (q\rho_0-(1-q)\rho_1)M_x\bigr| \right) \\
                      &=: \frac12\bigl( 1 - \| q\rho_0 - (1-q)\rho_1\|_M \bigr).
\end{split}\]
Observe that $\| \Delta \|_M = \sum_{x\in\cX} |\tr \Delta M_x|$ is a seminorm:
it is non-negative, homogeneous and obeys the triangle inequality. However, it may vanish on
$\Delta \neq 0$. This is excluded if the measurement $M$ is \textit{informationally
complete}, meaning that the operators $M_x$ span all the operators
over the Hilbert space: $\operatorname{span}\{ M_x:x\in\cX \} = \cB(\cH)$.

If not one but a whole set $\mathbf{M}$ of measurements is given, from which the
experimenter may choose, we have an equally natural (semi-)norm
\[
  \| \Delta \|_{\mathbf{M}} = \sup_{M\in\mathbf{M}} \| \Delta \|_M,
\]
in terms of which the minimum error probability is expressed
as $\frac12\bigl( 1 - \| q\rho_0 - (1-q)\rho_1\|_{\mathbf{M}} \bigr)$.
These norms, under certain restrictions of interest on the measurement,
will be the object of study in the present paper, and in particular
their comparison with the trace norm, which by a classic observation
of Holevo~\cite{Holevo} and Helstrom~\cite{Helstrom} equals the distinguishability norm
under the set of all possible measurements:
\[
  \| \Delta \|_{\mathbf{ALL}} = \sup_{M\text{ any POVM}} \| \Delta \|_M
                      = \| \Delta \|_1
                      = \tr |\Delta|.
\]

In this spirit, we continue an investigation begun in~\cite{MWW}, addressing
some of the questions left open there. The reader is referred to that
paper for further information about distinguishability norms
and their interpretation in terms of the geometry of certain convex bodies
of operators. Note however that many results from~\cite{MWW} are restricted
to traceless operators $\Delta = \frac12(\rho_0-\rho_1)$, corresponding
to equal prior probabilities $q=1-q=\frac12$. Of course, mathematically
and also in view of applications with unequal prior probabilities,
it makes sense to lift this restriction.

\medskip
The structure of the rest of the paper is as follows:
In section~\ref{sec:t-designs} we define the measurements and some
classes of measurements we will be interested in, introducing also a
multi-partite generalisation of the non-commutative $\ell_2$-norm
(aka Hilbert-Schmidt norm), denoted $\|\cdot\|_{2(K)}$.
In section~\ref{sec:2-norms} we then state and prove our main results
comparing measurement norms with $2$-norms, while in section~\ref{sec:1-norm}
we move on to relations with the trace norm and the application of
our results to so-called \emph{data hiding}.
We conclude in section~\ref{sec:conclusion} with a brief discussion.
Appendix~\ref{app:4-design-proofs} is devoted to the technical parts of
the proof of the main result, building on ideas from~\cite{AmbainisEmerson,MWW}.

\section{Projective $\mathbf{t}$-designs and LOCC measurements}
\label{sec:t-designs}
As explained in the introduction, if we want to use a single measurement
to define a norm it has to be \textit{informationally complete}. Among those,
there are measurements with special symmetry properties known as
\emph{(projective) designs} -- see~\cite{zauner,t-designs1,t-designs2}
and~\cite{AmbainisEmerson}.

\begin{definition}
  \label{def:t-design-POVM}
  A rank-one POVM $M=(M_x)_{x\in\cX}$ on a $d$-dimensional Hilbert space $\cH$
  is called a \emph{$t$-design} if
  the ensemble $\{ p_x, P_x \}$ of rank-one projectors, with
  $p_x = \frac{1}{d}\tr M_x$ and $P_x = \frac{M_x}{\tr M_x}$, is a
  projective (weighted) $t$-design in the usual sense \cite{zauner,t-designs1,t-designs2},
  i.e.~if
  \[
    \sum_{x\in\cX} p_x P_x^{\otimes t} = \int {\rm d}\psi\ \ket{\psi}\bra{\psi}^{\otimes t}.
  \]
  where the integral is over the uniform (unitary invariant) probability measure
  on the pure states of $\cH$.

  Note that
  \[
    \int {\rm d}\psi\ \ket{\psi}\bra{\psi}^{\otimes t} = \frac{1}{{d+t-1 \choose t}} \Pi_{\Sym}
                                       = \frac{1}{d(d+1)\cdots(d+t-1)} \sum_{\pi\in\mathfrak{S}_t} U_\pi,
  \]
  where $\Pi_{\Sym}$ is the projector onto the completely symmetric subspace of $\cH^{\otimes t}$,
  i.e. the subspace of $\cH^{\otimes t}$
  invariant under all the permutation unitaries
  $U_\pi \ket{v_1}\otimes\cdots\otimes\ket{v_t}
      = \ket{v_{\pi^{-1}(1)}}\otimes\cdots\otimes\ket{v_{\pi^{-1}(t)}}$.
\end{definition}

We shall be concerned with multi-partite quantum systems. To fix notation
for the rest of the paper,
let $\cH_1,...,\cH_K$ be $K$ finite dimensional Hilbert spaces
(with dimensions $d_j:=\dim \cH_j < \infty$), and $\cH=\cH_1\otimes \cdots \otimes \cH_K$
their tensor product (of dimension $D:=d_1\cdots d_K$).
Let furthermore $\Delta$ be a Hermitian operator on $\cH$.
The first measurements we shall be interested in, are tensor products
of $t$-designs: $M=M^{(1)} \otimes \cdots \otimes M^{(K)}$, where each
$M^{(j)}$ is a $t$-design POVM. In other words, the individual
elements of the partition of unity are all possible tensor products
$M^{(1)}_{x_1} \otimes \cdots \otimes M^{(K)}_{x_K}$. The following
observation makes it possible to use probabilistic techniques to analyse
the norm associated to a single measurement, paving the way to an
analysis of $t$-design measurements.

\begin{observation}
  \label{obs:POVM-bias}
  For a rank-one POVM $M=(M_x)_{x\in\cX}$ on a $d$-dimensional Hilbert space
  $\cH$, let $p_x = \frac{1}{d}\tr M_x$ and $P_x = \frac{M_x}{\tr M_x}$. Then,
  introducing a random index $X$ with $\Pr\{X=x\}=p_x$, $P_X$ is a random rank-one
  projector with expectation
  $\displaystyle{\EE P_X = \sum_{x\in\cX} p_x P_x = \frac{1}{d}\1}$. Furthermore,
  \[
    \|\Delta\|_M = d\,\EE|S|,
  \]
  for the real random variable $S = \tr\Delta P_X$.
  Indeed,
  \begin{equation*}
    d\,\EE|S| = d\sum_{x\in\cX} p_x |\tr \Delta P_x| = \sum_{x\in\cX} |\tr \Delta M_x| = \|\Delta\|_M.
  \end{equation*}
\end{observation}

Beyond these $t$-design tensor products, we are going to consider the class
of all POVMs implementable by a protocol of local operations and
classical communication (LOCC) which includes the above;
the class SEP consisting of all POVMs $M=(M_x)_{x\in\cX}$ with fully
separable operators $M_x \geq 0$,
which in turn contains LOCC; and finally the even larger class PPT
that is defined by $M_x^{\Gamma_I} \geq 0$ for all $x\in\cX$ and
$I\subset[K]$, where $\Gamma_I$ is the partial transpose on all
parties $I$.
By the definition of these classes, it is enough to consider two-outcome
POVMs $(M,\1-M)$ that can be implemented by LOCC, or such that both
$M$ and $\1-M$ are separable, or PPT with respect to all bipartite
cuts, respectively.
See~\cite{MWW} for a more detailed discussion of these classes.

\begin{definition}
  \label{def:k-party-2-norm}
  For any operator $\Delta$ (we only consider Hermitian ones in the
  following) on $\cH = \cH_1\otimes \cdots \otimes \cH_K$, let
  \[
    \| \Delta \|_{2(K)} := \sqrt{\sum_{I\subset[K]} \tr|\tr_I\Delta|^2},
  \]
  where $\tr_I$ denotes the partial trace over all parties $I$.

  Note that for $K=1$, this is ``almost'' the non-commutative $\ell_2$-norm:
  $\|\Delta\|_{2(1)} = \sqrt{|\tr\Delta|^2+\tr|\Delta|^2}$, reducing
  to the latter (aka Hilbert-Schmidt norm), $\|\Delta\|_2 = \sqrt{\tr\Delta\Delta^\dagger}$,
  on traceless operators.
\end{definition}

In \cite{MWW} such measurements and the above classes LOCC, SEP and PPT
were investigated in the case of $K=1$ and $K=2$ parties. Measurement norm
and $2$-norm were first directly related in \cite{AmbainisEmerson}, with
an application in quantum algorithms, while
Harrow \emph{et al.}~\cite{HarrowMontanaroShort}
were the first to realise that for a $4$-design POVM $M$,
the measurement norm and $\ell_2$-norm are indeed equivalent,
although only for traceless operators $\Delta$ on a single system:
\[
  \|\Delta\|_2 \geq \|\Delta\|_M \geq \frac{1}{3}\|\Delta\|_2,
\]
The extension to two parties in~\cite{MWW},
\[
  \|\Delta\|_M \geq \frac{1}{\sqrt{153}}\|\Delta\|_2,
\]
for a tensor product of two $4$-design POVMs and still assuming
$\tr\Delta = 0$, subsequently found applications in entanglement
theory~\cite{BrandaoChristandlYard:sq}, suggesting that our results
for larger $K$ might be useful, too.

\section{Comparison with $2$-norms}
\label{sec:2-norms}
Our first two theorems show that the norms related to $2$- and $4$-designs
are closely related to the norm $\|\cdot\|_{2(K)}$.

\begin{theorem}
  \label{thm:2-designs-times-k}
  If $M$ is a tensor product of $K$ $2$-design POVMs, then
  \[
    \| \Delta \|_{M} \leq \sqrt{\prod_{j=1}^K \frac{d_j}{d_j+1}} \| \Delta \|_{2(K)}
                             \leq \| \Delta \|_{2(K)}.
  \]
\end{theorem}

\begin{proof}
  Starting from observation~\ref{obs:POVM-bias}, with the random
  variable $S$ that takes the value $\tr\Delta P_{\underline{x}}$ with probability
  $p_{\underline{x}} = p_{x_1}\cdots p_{x_K}$,
  we have by the convexity of the square function,
  \begin{equation*}
    \|\Delta\|_M = D\EE|S| \leq D\sqrt{\EE S^2}.
  \end{equation*}
  Furthermore, using the definition of $2$-design,
  \[\begin{split}
    \EE S^2 &= \sum_{\underline{x}} p_{\underline{x}} \left( \tr\Delta P_{\underline{x}} \right)^2 \\
      &= \sum_{\underline{x}} p_{\underline{x}} \tr(\Delta\ox\Delta)(P_{\underline{x}}\ox P_{\underline{x}}) \\
      &= \tr\left( \Delta^{\ox 2} \bigotimes_{j=1}^K \frac{\1+F}{d_j(d_j+1)} \right) \\
      &= \prod_{j=1}^K \frac{1}{d_j(d_j+1)} \sum_{I\subset[K]} \tr\left(\tr_I \Delta\right)^2,
  \end{split}\]
  with the notation $F:=U_{(12)}$, and using the fact that $\tr (A\ox B)F = \tr AB$.
  Inserting this into the above inequality concludes the proof.
\end{proof}

\begin{theorem}
  \label{thm:4-designs-times-k}
  If $M$ is a tensor product of $K$ $4$-design POVMs, then
  \[
    \sqrt{\frac{1}{18}}^{K} \| \Delta \|_{2(K)} \leq \| \Delta \|_{M} \leq \| \Delta \|_{2(K)}.
  \]
\end{theorem}

\begin{proof}
  Again we start with observation~\ref{obs:POVM-bias}, with the random
  variable $S$ that takes the value $\tr\Delta P_{\underline{x}}$ with probability
  $p_{\underline{x}} = p_{x_1}\cdots p_{x_K}$.

  The upper bound is contained in theorem~\ref{thm:2-designs-times-k}, as a
  $t$-design is automatically a $(t-1)$-design.
  For the lower bound, we
  follow the strategy of Ambainis and Emerson~\cite{AmbainisEmerson},
  using this inequality of Berger's~\cite{Berger:inequality} (by the
  way a special case of H\"older's inequality):
  \[
    \EE|S| \geq \sqrt{\frac{(\EE S^2)^3}{\EE S^4}}.
  \]
  In the proof of theorem~\ref{thm:2-designs-times-k} we have already
  calculated
  \[
    \EE S^2 = \prod_{j=1}^K \frac{1}{d_j(d_j+1)} \sum_{I\subset[K]} \tr\left(\tr_I \Delta\right)^2.
  \]
  Using the property of $4$-design, we similarly get
  \[
    \EE S^4 = \prod_{j=1}^K \frac{1}{d_j(d_j+1)(d_j+2)(d_j+3)}
              \tr\left( \Delta^{\ox 4}
                        \left( \sum_{\underline{\pi}\in\mathfrak{S}_4^K} U_{\underline{\pi}} \right)
                  \right),
  \]
  with the notation $\displaystyle{U_{\underline{\pi}}:=\bigotimes_{j=1}^KU_{\pi_j}}$
  for $\underline{\pi}=(\pi_1,\ldots,\pi_K)$.

  Thus it suffices to show
  \begin{equation*}
    \tr\left( \Delta^{\ox 4}
              \left( \sum_{\underline{\pi}\in\mathfrak{S}_4^K} U_{\underline{\pi}} \right)
        \right)              \leq 18^K \left[ \sum_{I\subset[K]} \tr\left(\tr_I \Delta\right)^2 \right]^2,
  \end{equation*}
  which is precisely proposition~\ref{prop:4th-vs-2nd-moment} in
  appendix~\ref{app:4-design-proofs}, and we are done.
\end{proof}

\bigskip\noindent
\emph{Alternative proof of a weaker version of theorem~\ref{thm:4-designs-times-k}.}
  Here is a way of demonstrating the slightly worse bound
  \begin{equation*}
    \tr\left( \Delta^{\ox 4}
              \left( \sum_{\underline{\pi}\in\mathfrak{S}_4^K} U_{\underline{\pi}} \right)
        \right)            \leq 24^K \left[ \sum_{I\subset[K]} \tr\left(\tr_I \Delta\right)^2 \right]^2,
  \end{equation*}
  which has the advantage of being conceptually simple, and showing some
  of the tricks used in the proof of proposition~\ref{prop:4th-vs-2nd-moment}.
  For this it is enough to show that, for every $K$-tuple $\underline{\pi}\in \mathfrak{S}_4^K$:
  \begin{equation}
    \label{eq:simpler-task}
    t({\underline{\pi}}) :=   \left| \tr \Delta^{\ox 4} U_{\underline{\pi}} \right|
                        \leq \max_{I\subset[K]} \left[ \tr\left(\tr_I \Delta\right)^2 \right]^2.
  \end{equation}

  The basic idea is to use Cauchy-Schwarz inequality
  as in~\cite{MWW}, but now repeatedly:
  For arbitrary (compatible) operators $X$ and $Y$,
  \[
    \big|\tr XY^\dagger\big| \leq \sqrt{(\tr XX^\dagger)(\tr YY^\dagger)}.
  \]

  Concretely, given Hermitian operators $M_1$, $M_2$, $M_3$, $M_4$ on a Hilbert
  space $\mathcal{K}$ and a permutation $\sigma\in\mathfrak{S}_4$ with corresponding
  unitary $U_{\sigma}$ on $\mathcal{K}^{\otimes 4}$, we may write
  \[
    \tr U_{\sigma}(M_1\otimes M_2\otimes M_3\otimes M_4) = \tr XY^{\dagger},
  \]
  with operators $X$ and $Y$ mapping $\mathcal{K}^{\otimes k}$ to
  $\mathcal{K}^{\otimes \ell}$, where $k$ and $\ell$ depend on the permutation $\sigma$,
  and may well be $0$.
  To be precise, taking the $M_a$ as matrices, the left hand trace above is a
  contraction of the eight (four upper and four lower) indices of
  $M_1\otimes M_2\otimes M_3\otimes M_4$, where $\sigma$ tell us which
  pairs are to be contracted, namely the upper index of $M_a$ with the lower
  index of $M_{\sigma(a)}$. Now, $X$ is the tensor contraction of the part
  of this network involving $M_1$ and $M_2$, and $Y^\dagger$ is the contraction
  of the remaining part, involving $M_3$ and $M_4$;
  note that $X$ and $Y$ may be numbers, but
  usually are matrices because of the ``dangling'' indices connecting them.
  With this,
  \begin{align*}
    \tr XX^{\dagger} &= \tr U_{\sigma^L}(M_1\otimes M_2\otimes M_2\otimes M_1), \\
    \tr YY^{\dagger} &= \tr U_{\sigma^R}(M_4\otimes M_3\otimes M_3\otimes M_4).
  \end{align*}
  with certain permutations $\sigma^L,\sigma^R\in\mathfrak{S}_4$.

  What we gain by doing so is that $\sigma^L$ and $\sigma^R$ are not arbitrary
  elements of $\mathfrak{S}_4$; rather, they necessarily belong to the subset
  $\mathfrak{A}:=\{\id,(14),(23),(1234),(1432),(12)(34),(14)(23)\}$,
  which is stable under the exchange $1\leftrightarrow 4$ and $2\leftrightarrow 3$,
  i.e. under the conjugation by $(14)(23)$.
  In order to easily see into which pair $(\sigma^L,\sigma^R)\in\mathfrak{A}\times\mathfrak{A}$
  each permutation $\sigma\in\mathfrak{S}_4$ splits, we can make use of Penrose's
  tensor diagrams~\cite{Penrose}, which we briefly explain here.

  For any Hermitian $M$ on $\mathcal{K}$ and unit vectors
  $|i\rangle$, $|j\rangle\in\mathcal{K}$ we represent the matrix element
  $\langle i|M|j\rangle$ by the following diagram with terminals:
  \[
    \includegraphics[width=2cm]{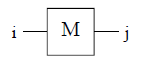}
  \]
  Summing matrix elements over an orthonormal basis of $\mathcal{K}$
  is represented by joining the corresponding terminals. So, for instance,
  $\tr M=\displaystyle{\sum_j \bra{j} M \ket{j}}$ is represented by
  \[
    \includegraphics[width=2cm]{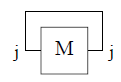}
  \]
  And in the same way for matrix multiplication,
  $\displaystyle{\bra{i}MN\ket{k}=\sum_j \bra{i}M\proj{j}N\ket{k}}$ is
  represented by
  \[
    \includegraphics[width=3cm]{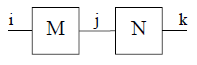}
  \]
  The expressions we looked at above are, for Hermitian $M_1$, $M_2$, $M_3$, $M_4$
  on $\mathcal{K}$ and $\sigma\in\mathfrak{S}_4$:
  \[
    \tr U_{\sigma}(M_1\otimes M_2\otimes M_3\otimes M_4)
       =\sum_{i_1,i_2,i_3,i_4} \langle i_1|M_1|i_{\sigma(1)}\rangle
                               \langle i_2|M_2|i_{\sigma(2)}\rangle
                               \langle i_3|M_3|i_{\sigma(3)}\rangle
                               \langle i_4|M_4|i_{\sigma(4)}\rangle.
  \]
  For instance, $\tr U_{(123)}(M_1\otimes M_2\otimes M_3\otimes M_4)$ is represented by
  \[
    \includegraphics[width=5.5cm]{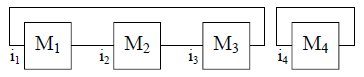}
  \]
  and $X = M_1 M_2$, $Y^\dagger=M_3 (\tr M_4)$.

  In this case, the splitting procedure and use of Cauchy-Schwarz described above
  can be diagrammatically written as
  \[
    \vrule\ \includegraphics[width=5cm]{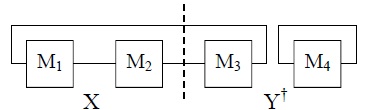} \vrule\
        \stackrel{\leq}{\phantom{=}}
     \sqrt{\includegraphics[width=10cm]{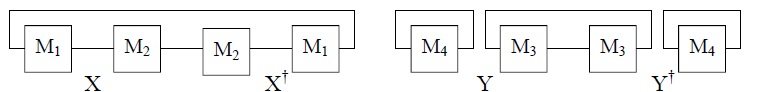}}
  \]
  which means that for $\sigma=(123)$, we have $\sigma^L=(1234)$ and $\sigma^R=(23)$.

  The resulting splitting map
  $\text{Split}:\mathfrak{S}_4\ni\sigma\mapsto(\sigma^L,\sigma^R)\in\mathfrak{A}\times\mathfrak{A}$
  for each $\sigma\in\mathfrak{S}_4$ can then easily be constructed and looked
  up in the table of Fig.~\ref{fig:split-table}.

  \begin{figure}[ht]
  \begin{tabular}{|l||c||c|c|}\hline
   Conj.~class & $\sigma$    & $\sigma^L$ & $\sigma^R$ \\ \hline\hline
   1111        & $\id$    & $\id$      & $\id$      \\ \hline\hline
   211         & (12)     & (12)(34)   & $\id$      \\ \hline
               & (13)     & (14)       & (23)       \\ \hline
               & (14)     & (14)       & (14)       \\ \hline
               & (23)     & (23)       & (23)       \\ \hline
               & (24)     & (23)       & (14)       \\ \hline
               & (34)     & $\id$      & (12)(34)       \\ \hline\hline
   22          & ~(12)(34)~ & ~(12)(34)~   & ~(12)(34)~   \\ \hline
               & (13)(24) & (14)(23)   & (14)(23)   \\ \hline
               & (14)(23) & (14)(23)   & (14)(23)   \\ \hline\hline
   31          & (123)    & (1234)     & (23)       \\ \hline
               & (132)    & (1432)     & (23)       \\ \hline
               & (124)    & (1234)     & (14)       \\ \hline
               & (142)    & (1432)     & (14)       \\ \hline
               & (134)    & (14)       & (1234)     \\ \hline
               & (143)    & (14)       & (1432)     \\ \hline
               & (234)    & (23)       & (1234)     \\ \hline
               & (243)    & (23)       & (1432)     \\ \hline\hline
   4           & (1234)   & (1234)     & (1234)     \\ \hline
               & (1243)   & (1234)     & (1432)     \\ \hline
               & (1324)   & (14)(23)   & (14)(23)   \\ \hline
               & (1342)   & (1432)     & (1234)     \\ \hline
               & (1432)   & (1432)     & (1432)     \\ \hline
               & (1423)   & (14)(23)   & (14)(23)   \\ \hline
  \end{tabular}
  \caption{Table of the splitting map
     $\text{Split}:\mathfrak{S}_4 \longrightarrow \mathfrak{A} \times \mathfrak{A}$,
     $\text{Split}(\sigma) = (\sigma^L,\sigma^R)$, grouped according to conjugacy classes of $\sigma$.}
  \label{fig:split-table}
  \end{figure}

  Trivially extending the above reasoning to $K$-tuples
  $\underline{\pi}=(\pi_1,\ldots,\pi_K) \in\mathfrak{S}_4^K$ of permutations, we
  apply the splitting map to all the $\pi_i$ ($1\leq i\leq K$), and use the
  Cauchy-Schwarz  as well as geometric-arithmetic mean inequality:
  \begin{equation}\begin{split}
    \label{eq:CS}
     t(\underline{\pi})
         =     \left|\tr \Delta^{\ox 4}U_{\underline{\pi}} \right|
         \leq \sqrt{\left( \tr \Delta^{\ox 4}U_{\underline{\pi}^L} \right)
                    \left( \tr \Delta^{\ox 4}U_{\underline{\pi}^R} \right)}
        &=    \sqrt{t({\underline{\pi}}^L) t({\underline{\pi}}^R)} \\
        &\leq \frac12 t({\underline{\pi}}^L) + \frac12 t({\underline{\pi}}^R).
  \end{split}\end{equation}

  The other observation we use is that $t(\underline{\pi})$ is invariant
  under conjugation by elements from the diagonal subgroup
  $\mathfrak{G}:=\{(\sigma,\ldots,\sigma),\ \sigma\in\mathfrak{S}_4\}$
  of $\mathfrak{S}_4^K$, because $\Delta^{\otimes4}$ is invariant under conjugation
  by elements of the form $(U_{\sigma})^{\otimes K}$.
  Now, notice that the subset $\mathfrak{A}_0 = \{\id,(12)(34),(14)(23)\}$ of $\mathfrak{A}$
  is such that $\mathfrak{A}_0^K$ is stable under conjugation by any element of $\mathfrak{G}$
  followed by splitting. And what is more, any given $\underline{\pi}\in\mathfrak{S}_4^K$
  can be transformed into a family of elements of $\mathfrak{A}_0$ by repeatedly conjugating
  by elements of $\mathfrak{G}$ and splitting.

  Thus, using eq.~(\ref{eq:CS}) and conjugation invariance repeatedly, we eventually
  get for all $\underline{\pi}\in\mathfrak{S}_4^K$ the upper bound
  \[
    t({\underline{\pi}}) \leq \sum_\alpha p_\alpha t(\underline{\pi}^{(\alpha)}),
  \]
  with certain $p_\alpha\geq 0$ summing to $1$, and
  $\underline{\pi}^{(\alpha)}$ belonging to $\mathfrak{A}_0^K$. As we have
  not attempted to control the coefficients $p_\alpha$, we record as a useful
  intermediate bound for all $\underline{\pi}\in\mathfrak{S}_4^K$,
  \begin{equation}
    \label{eq:bound-t}
    t({\underline{\pi}}) \leq \max_{\underline{\sigma}\in\cA_0^K} t(\underline{\sigma}).
  \end{equation}
  As a matter of fact, we know already how to upper bound the traces on the right hand side
  of eq.~(\ref{eq:bound-t}). Indeed, a generic $\underline{\sigma}\in\cA_0^K$ is given by
  disjoint subsets $I,J\subset [K]$, such that:
  \[
    {\sigma}_j = \begin{cases}
                   \id      & j \in I, \\
                   (12)(34) & j \in J, \\
                   (14)(23) & j \in J':=[K]\setminus (I\cup J).
                 \end{cases}
  \]
  Hence,
  \[
    t(\underline{\sigma}) = \tr \Delta^{\ox 4} U_{\underline{\sigma}}
                          = \tr\left( \bigl(\tr_I\Delta\bigr)^{\ox 4}
                                      \bigl(U_{(12)(34)}^{\ox J} \ox U_{(14)(23)}^{\ox J'} \bigr) \right),
  \]
  which really is a bipartite term (i.e., $K=2$) as treated
  in~\cite[Proof of Lemma 26, case ``(2,2):(2,2)'']{MWW}:
  there it was shown to be $\leq \left[ \tr(\tr_I\Delta)^2 \right]^2$.

  Combining this bound with eq.~(\ref{eq:bound-t}), we obtain eq.~(\ref{eq:simpler-task}),
  and we are done.
\qed

\bigskip
Theorem~\ref{thm:4-designs-times-k} extends the results of~\cite{AmbainisEmerson,MWW}
to $K>2$; however we may wonder how good the lower bound really is, and in particular
if the dependence on $K$ is ``real''. The following result shows that, indeed, the constant
relating $\|\Delta\|_M$ has to decrease as a power of $K$. For this it is enough to
analyse a specific tensor product of local $4$-design POVMs, and we choose
$U_{\cH}:=U_{\cH_1}\otimes\cdots\otimes U_{\cH_K}$, the tensor product of
the $K$ uniform (unitary invariant) POVMs on sub-systems $\cH_j$ ($j=1,\ldots,K$).
This is an interesting measurement since each of the $U_{\cH_j}$ is an $\infty$-design,
in particular a $4$-design, and we can exploit the symmetry to make calculations
feasible.
Whereas theorem~\ref{thm:4-designs-times-k} gives us
\[
   \sqrt{\frac{1}{18}}^K \| \Delta \|_{2}
        \leq \sqrt{\frac{1}{18}}^K \| \Delta \|_{2(K)}
        \leq \| \Delta \|_{U_{\cH}}
        \leq \| \Delta \|_{2(K)},
\]
we have the following:

\begin{proposition}
  \label{prop:u-tight}
  There exists a Hermitian $\Delta\neq 0$ on $\cH$ such that
  \[
    \|\Delta\|_{U_{\cH}} = \sqrt{\frac12}^K \|\Delta\|_{2(K)} = \sqrt{\frac12}^K \|\Delta\|_{2}.
  \]
\end{proposition}
\begin{proof}
  Define $\Delta_j = \frac12\proj{0} - \frac12\proj{1}$ with orthogonal unit
  vectors $\ket{0},\ket{1} \in \cH_j$, $j=1,\ldots,K$.
  Let $\Delta = \bigotimes_{j=1}^K \Delta_j$
  Clearly, $\tr\Delta_j = 0$ and $\|\Delta_j\|_2 = \sqrt{\frac12}$ for all $j$, while
  from~\cite[Theorem 10]{MWW} we know that $\|\Delta_j\|_{U_{\cH_j}} = \frac12$.

  Hence, $\|\Delta\|_{2(K)} = \|\Delta\|_2 = 2^{-\frac{K}{2}}$; on the
  other hand, exploiting the tensor product structure of both state and
  measurement, $\|\Delta\|_{U_{\cH}} = 2^{-K}$.
\end{proof}

\bigskip
We shall now move on to investigating the properties of the measurement norms associated
with not one but a whole class of locally restricted measurements.
\begin{theorem}
  \label{thm:sep-vs-2}
  \label{thm:ppt-vs-2}
  For any number $K\geq 2$ of parties and any local dimensions $d_j$ ($1\leq j\leq K$),
  \[
    \|\Delta\|_{\mathbf{SEP}} \geq 2\sqrt{\frac12}^K \|\Delta\|_2.
  \]
  Furthermore,
  \[
    \|\Delta\|_{\mathbf{PPT}} \geq \|\Delta\|_2.
  \]
\end{theorem}
\begin{proof}
  The first inequality above was already shown in~\cite{MWW}, but we repeat the
  proof since it is very simple:
  It uses a result of Barnum and Gurvits~\cite{BarnumGurvits},
  that for any Hermitian $X$ on a $K$-partite Hilbert space, if
  $\|X-\1\|_2 \leq 2^{1-\frac{K}{2}}$, then $X$ is separable.

  So, $\|2M-\1\|_2\leq 2^{1-\frac{K}{2}}$ implies that both
  $\1+(2M-\1)=2M$ and $\1+(\1-2M)=2(\1-M)$
  are separable, i.e.~$M$ and $\1-M$ are separable operators.
  Thus, for our Hermitian $\Delta$ on $\cH$:
  \[\begin{split}
    \|\Delta\|_{\mathbf{SEP}} &=    \max_{(M,\1-M)\in\mathbf{SEP}} \bigl| \tr\Delta(2M-\1) \bigr| \\
                              &\geq \max_{\|A\|_2 \leq 2^{1-\frac{K}{2}}} |\tr \Delta A| \\
                              &=    2^{1-\frac{K}{2}}\|\Delta\|_2,
  \end{split}\]
  where the last equality is by self-duality of the $\ell_2$-norm.

  \medskip
  To show the second inequality, notice that
  $(M,\1-M)$ being a two-outcome PPT POVM is a consequence of $M$ and $\1-M$
  being separable for any bipartition of the $K$ parties.

  Thus, we can use once more the Barnum-Gurvits result~\cite{BarnumGurvits}:
  If $\|2M-\1\|_2\leq 1$, then $M$ and $\1-M$ are both separable with respect to
  any bipartition, hence PPT with respect to any bipartition.
  The claim follows now as in the first part.
\end{proof}

\section{Comparison with trace norm and data hiding}
\label{sec:1-norm}
All measurement norms are trivially upper bounded by the trace norm
$\|\cdot\|_1$. In the other direction, the standard
$\|\Delta\|_1 \leq \sqrt{D}\|\Delta\|_2$ for operators $\Delta$ on a
$D$-dimensional Hilbert space, allows us to turn the $\ell_2$-norm
estimates from the previous section into lower bounds on
$\|\Delta\|_M$, which in turn provides a lower bound on
$\|\Delta\|_{\mathbf{LOCC}}
   \leq \|\Delta\|_{\mathbf{SEP}}
   \leq \|\Delta\|_{\mathbf{PPT}}$:
\begin{equation}
  \label{eq:data-hiding-bound}
  \|\Delta\|_M \geq \sqrt{\frac{1}{18}}^K\frac{1}{\sqrt{D}}\|\Delta\|_1,
\end{equation}
for any tensor product of $4$-design POVMs, $M$. These are non-trivial because
by now it is a classic result in quantum information that quantum states
allow for \emph{data hiding}~\cite{DiVincenzoLeungTerhal}: namely, on large
composite systems there exist states with orthogonal supports (hence perfect
distinguishability by a suitable measurement) that are nevertheless barely distinguishable by LOCC.

More of this below, but let us start with some simple observations:
That both the occurrence of the inverse square root of $D$, and the
exponential dependence of the lower bound on $K$ are not artifacts,
is shown by the following example.

\begin{proposition}
  \label{prop:u-tight-tracenorm}
  Consider the measurement $U_{\cH}:=U_{\cH_1}\otimes\cdots\otimes U_{\cH_K}$,
  the tensor product of the $K$ uniform (unitary invariant) POVMs on sub-systems
  $\cH_j$ with dimensions $d_j$, $j=1,\ldots,K$.
  There exists a Hermitian $\Delta\neq 0$ such that
  \[
    \|\Delta\|_{U_{\cH}} \leq \left(\sqrt{\frac{2}{\pi}}+o(1)\right)^K \frac{1}{\sqrt{D}}\|\Delta\|_1,
  \]
  where $o(1)$ is arbitrarily small for sufficiently large
  $d_{\min} = \min\{d_1,\ldots,d_K\}$.
\end{proposition}
\begin{proof}
  Without loss of generality, all $d_j$ are even. Pick any projector $P_j$
  of rank $\frac{d_j}{2}$ in $\cH_j$, $Q_j:=\1-P_j$, and let
  $\Delta_j := \frac{1}{d_j}P_j - \frac{1}{d_j}Q_j$, so that $\|\Delta_j\|_1=1$.

  Our candidate is $\Delta = \bigotimes_{j=1}^K \Delta_j$, which also has trace norm
  $1$. On the other hand, by~\cite[Theorem 10]{MWW} we have
  \[
    \|\Delta_j\|_{U_{\cH_j}} \leq \left(\sqrt{\frac{2}{\pi}}+o(1)\right) \frac{1}{\sqrt{d_j}},
  \]
  and since both $\Delta$ and the measurement share the tensor product
  structure, we obtain the claim by multiplying together these inequalities.
\end{proof}

\bigskip
That the factor of $\frac{1}{\sqrt{D}}$ does not go away when we
go to the class of all LOCC, and indeed all PPT measurements, is
contained in the two following theorems.

\begin{theorem}
  \label{thm:data-hiding1}
  When all the local dimensions are $d$, hence $D=d^K$, there
  exists a traceless Hermitian $\Delta \neq 0$ with
  \[
    \|\Delta\|_{\mathbf{PPT}} \leq \frac{2}{d^{\left\lfloor K/2 \right\rfloor}-1}\|\Delta\|_1.
  \]
\end{theorem}

In other words, one can find two states $\rho_0$ and $\rho_1$ with orthogonal
supports (i.e., $\frac12\rho_0-\frac12\rho_1$ has trace norm $1$), such that
\[
  \left\|\frac12\rho_0 - \frac12\rho_1 \right\|_{\mathbf{PPT}}
                                       \leq \frac{2\sqrt{d}^{\kappa}}{\sqrt{D}-\sqrt{d}^{\kappa}} \leq\frac{3\sqrt{d}^{\kappa}}{\sqrt{D}},
\]
where $\kappa = K \!\!\!\mod 2$ is the parity of $K$.
Hence these two states are data hiding in the sense of~\cite{DiVincenzoLeungTerhal}:
$\rho_i$ encodes a state between $K$ parties, but as long as those
are restricted to LOCC measurements (or more generally PPT measurements),
they have only a very slim chance of guessing this state. Indeed, the
probability of discriminating correctly $\rho_0$ from $\rho_1$ decreases
as the inverse square root of the total dimension $D$, and
eq.~(\ref{eq:data-hiding-bound}) shows that this order
of magnitude is essentially optimal, apart from a $K$-dependent constant.

\medskip
\begin{proof}
For all Hermitian $\Delta$,
\[\begin{split}
  \|\Delta\|_{\mathbf{PPT}} &= \max_{\left(\frac12(\1+A),\frac12(\1-A)\right)\in \mathbf{ PPT}} |\tr A\Delta| \\
                            &= \max_{\forall I\subset[K],-\1\leq A^{\Gamma_I}\leq\1} |\tr A\Delta|.
\end{split}\]
Yet, if $A$ is such that for $I\subset[K]$, $-\1\leq A^{\Gamma_I}\leq\1$, then necessarily
\[
  |\tr A\Delta| =    \left|\tr A^{\Gamma_I}\Delta^{\Gamma_I}\right|
                \leq \left\| A^{\Gamma_I} \right\|_{\infty} \left\|\Delta^{\Gamma_I}\right\|_1
                \leq \left\|\Delta^{\Gamma_I}\right\|_1.
\]
Among the operators for which we know how to evaluate the trace norm of any of their
partial transposes are the permutation operators $U_{\pi}$, $\pi\in\mathfrak{S}_K$.
Indeed, for all $I:=\{1,\ldots,p\}\subset[K]$ we have
\[
  U_{\pi}^{\Gamma_I}
     = \sum_{j_1,\ldots,j_K} \ket{j_{\pi(1)},\ldots,j_{\pi(p)},j_{p+1},\ldots,j_K}
                             \!\bra{j_1,\ldots,j_p,j_{\pi(p+1)},\ldots,j_{\pi(K)}}.
\]
Hence, letting $f(I,\pi):=|\{i\in I,\pi(i)\notin I\}|$, we get:
$\displaystyle{\left\|U_{\pi}^{\Gamma_I}\right\|_1 = d^{K-f(I,\pi)}}$.

Choosing as permutation $\pi$ the product of $\lfloor K/2\rfloor$ disjoint transpositions,
$\pi:=(1,\lfloor K/2\rfloor +1)\ldots(\lfloor K/2\rfloor,2\lfloor K/2\rfloor)$ (that decomposes therefore into $\lceil K/2\rceil$ disjoint cycles),
let us now consider the following traceless Hermitian $\Delta$:
\[\begin{split}
  \Delta &:= \frac{1}{d^K+d^{\lceil K/2\rceil}}(\1+U_\pi)-\frac{1}{d^K-d^{\lceil K/2\rceil}}(\1-U_\pi) \\
         &=  \frac{2}{d^{\lceil K/2\rceil}(d^{2\lfloor K/2\rfloor}-1)}
                 \left(d^{\lfloor K/2\rfloor}U_\pi-\1\right).
\end{split}\]
Note that $\Delta$ is the difference of the two orthogonal density operators
$\rho_0:=\frac{1}{d^K+d^{\lceil K/2\rceil}}(\1+U_\pi)$ and
$\rho_1:=\frac{1}{d^K-d^{\lceil K/2\rceil}}(\1-U_\pi)$, hence $\|\Delta\|_1=2$.

Furthermore, $I:=\{1,\ldots,\lfloor K/2\rfloor\}\subset[K]$ is such that
$f(I,\pi)=\lfloor K/2\rfloor$, so $\left\|U_\pi^{\Gamma_I}\right\|_1=d^{K-\lfloor K/2\rfloor}$,
and hence, after a straightforward calculation:
\[
  \left\|\Delta^{\Gamma_I}\right\|_1
    \leq \frac{2}{d^{\lceil K/2\rceil}(d^{2\lfloor K/2\rfloor}-1)}
         \left(d^{\lfloor K/2\rfloor}\left\|U_\pi^{\Gamma_I}\right\|_1+\left\|\1^{\Gamma_I}\right\|_1\right)
    \leq \frac{2}{d^{\lfloor K/2\rfloor}-1}\|\Delta\|_1.
\]
Thus, $\displaystyle{\|\Delta\|_{\mathbf{PPT}} \leq \left\|\Delta^{\Gamma_I} \right\|_1
                                               \leq\frac{2}{d^{\lfloor K/2\rfloor}-1}\|\Delta\|_1}$,
which is what we wanted to prove.
\end{proof}

\bigskip
Theorem~\ref{thm:data-hiding1} and eq.~(\ref{eq:data-hiding-bound}) show
that -- at least for even $K$ -- the best performance for $K$-party
data hiding is indeed a bias inversely proportional to the square root
of the dimension, with a constant factor only depending on $K$.
Here is another construction that works also for
odd number $K$ of parties, with possibly unequal local dimensions.

\begin{theorem}
  \label{thm:data-hiding2}
  When all the local dimensions $d_j$ ($1\leq j\leq K$) are such that there exists a
  $I\subset[K]$ such that ${\cA:=\bigotimes_{j\in I}\cH_j}$ and
  ${\cB:=\bigotimes_{j\in[K]\setminus I}\cH_j}$
  satisfy $\dim\cA=\dim\cB=\sqrt{D}$, then there exists a traceless Hermitian $\Delta \neq 0$ with
  \[
    \|\Delta\|_{\mathbf{PPT}} \leq \frac{2}{\sqrt{D}+1}\|\Delta\|_1.
  \]
\end{theorem}

\begin{proof}
Denoting by $F$ the swap operator between the Hilbert spaces $\cA$ and $\cB$,
we let $\sigma$ and $\alpha$ be the normalised projectors onto the symmetric and
antisymmetric subspaces of $\mathbb{C}^{\sqrt{D}}\otimes\mathbb{C}^{\sqrt{D}}$,
respectively: $\sigma:=\frac{1}{\sqrt{D}(\sqrt{D}+1)}(\1+F)$ and
$\alpha:=\frac{1}{\sqrt{D}(\sqrt{D}-1)}(\1-F)$.
We then consider the traceless Hermitian $\Delta:=\sigma-\alpha$.

Now, if a POVM is PPT across all possible bipartitions of $\cH$, it is in
particular PPT across the bipartition $\cA:\cB$. As a consequence,
\[
  \|\Delta\|_{\mathbf{PPT}} \leq \|\Delta\|_{\mathbf{PPT(\cA:\cB)}} = \frac{2}{\sqrt{D}+1}\|\Delta\|_1,
\]
where the last equality is the original quantum data hiding result,
as shown in~\cite{DiVincenzoLeungTerhal,EggelingWerner:data-hiding}.
\end{proof}

\section{Conclusion}
\label{sec:conclusion}
We have solved an open problem from~\cite{MWW}, showing that for any
number $K$ of parties, the measurement norm on Hermitian operators
defined by local $4$-designs is equivalent to a certain relative
of the Hilbert-Schmidt norm. The equivalence is in terms of constants
of domination which depend only on the number of parties, not on the
local dimensions.

\begin{figure}[h]
\fbox{
\begin{minipage}{13cm}
\begin{alignat*}{5}
\frac{1}{\sqrt{D}}\|\Delta\|_1 &\leq&
                       \|\Delta\|_2 &\leq&
                                        & &           \|\Delta\|_\PPT        &\leq& \|\Delta\|_1 \\
   & &
                           & &
                                        & &                 \vertleq\qquad              & & \\
2\sqrt{\frac12}^K\frac{1}{\sqrt{D}}\|\Delta\|_1 &\leq&
                       2\sqrt{\frac12}^K\|\Delta\|_2 &\leq&
                                        & &           \|\Delta\|_\SEP                & & \\
   & &
                           & &
                                        & &                 \vertleq\qquad              & & \\
   & &
                           & &
                                        & &              \|\Delta\|_\LOCC               & & \\
   & &
                           & &
                                        & &                 \vertleq\qquad              & & \\
\sqrt{\frac{1}{18}}^K\frac{1}{\sqrt{D}}\|\Delta\|_1 &\leq&
    \sqrt{\frac{1}{18}}^K \|\Delta\|_2 &\leq& \ \sqrt{\frac{1}{18}}^K \|\Delta\|_{2(K)}
                                        &\leq&         \ \|\Delta\|_{M_4^{(K)}} &\leq & \ \|\Delta\|_{2(K)}
\end{alignat*}
\vspace{1mm}
\end{minipage}
}
\caption{A schematic summary of the new and previously known relations for any Hermitian $\Delta$; $M_4^{(K)}$ denotes
         a generic tensor product of $K$ $4$-design POVMs.}
\label{fig:schematic}
\end{figure}

Note that our constants appear worse compared to the known inequalities
for $K=1$ and $K=2$: In the former case, \cite{AmbainisEmerson} gives
$\frac13$ whereas we get $\sqrt{\frac{1}{18}}$; in the latter, \cite{MWW}
gives $\frac{1}{\sqrt{153}}$ whereas we get $\frac{1}{18}$. While the
gap is small, it may to some degree be explained by the fact that
in both these cited papers the assumption $\tr\Delta = 0$ was made, and
exploited to simplify the fourth moment even more. We believe that there is
merit in transcending this restriction, as not in all applications
it can be justified.
In any case, we leave it as an open problem to find the optimal constants
of domination with respect to the $\|\cdot\|_{2(K)}$ norm.

On a single system, the relation between measurement norm
and $2$-norm was exploited in \cite{AmbainisEmerson} to show that
even approximate $4$-design POVMs are derandomizing,
with an application in quantum algorithms.
The extension to bi-partite systems subsequently found applications in entanglement
theory in \cite{BrandaoChristandlYard:sq} where it was used to describe an
algorithm that would decide in a quasipolynomial time whether a bipartite
state is separable or whether it is ``far away'' from the set of separable states.
This suggests that our generalised results,
for any number of parties and non necesarily traceless Hermitians, might be useful too.
Indeed, since the first circulation of this paper as a preprint, our bounds have
been applied in precisely such a context~\cite{BrandaoHarrow}.
We use the occasion to draw attention also to Montanaro's subsequent
paper~\cite{Montanaro}, in which he derives bounds quite similar to ours,
although by a completely different approach, namely using hypercontractive inequalities.
As these are much more general and powerful tools, the actual constants
obtained by Montanaro are worse than the present ones out of elementary reasoning,
but it is pleasing to note that our results, rather than being an ``accident''
or just the consequence of a trick, have a firm structural basis.

Via the non-commutative $\ell_2$-norm we then obtained
performance comparisons with the trace norm, revealing
at most a factor of the order of the inverse square root of the
dimension of the total Hilbert space between the measurement norm and
the trace norm. Since the measurement is a
particular LOCC strategy, we get lower bounds on the distinguishing
power of LOCC measurements. The bounds can be shown to be optimal
in their dimensional dependence, as we exhibited two constructions of
data hiding states which attain these bounds up to $K$-dependent
factors.

\begin{figure}[h]
\fbox{
\begin{minipage}{9cm}
\begin{alignat*}{5}
\exists\ \Delta\neq0:\ &\|\Delta\|_\PPT& &\leq\ \frac{2}{\sqrt{D}+1}\|\Delta\|_1& \\
\exists\ \Delta\neq0:\ &\|\Delta\|_{U^{(K)}}& &\leq\ \left(\sqrt{\frac{2}{\pi}}+\delta\right)^K\frac{1}{\sqrt{D}}\|\Delta\|_1& \\
\exists\ \Delta\neq0:\ &\|\Delta\|_{U^{(K)}}& &\leq\ \sqrt{\frac{1}{2}}^K\|\Delta\|_{2(K)}=\sqrt{\frac{1}{2}}^K\|\Delta\|_2&
\end{alignat*}
\vspace{1mm}
\end{minipage}
}
\caption{A schematic summary of some of the tightness results obtained for the lower bounds; $U^{(K)}$ denotes
         the tensor product of the $K$ uniform POVMs.}
\label{fig:schematicbis}
\end{figure}

Here, one remaining question is whether for odd
number $K$ of parties, all of which have equal dimension, the additional
factor of square root of the local dimension can be removed in
theorem~\ref{thm:data-hiding1}. On a related note, with respect to
theorem~\ref{thm:data-hiding2}, does there exist a universal constant $C>0$
such that for all sufficiently large $D$ one can find Hermitian
$\Delta \neq 0$ with
$\| \Delta \|_{\mathbf{PPT}} \leq \frac{C}{\sqrt{D}} \| \Delta \|_1$,
irrespective of the local dimensions?

Even more interesting would be to quantify the performance of LOCC,
or at least fully separable (SEP), measurements: Indeed, notice that
in theorems~\ref{thm:data-hiding1} and \ref{thm:data-hiding2}, we have
only exploited \emph{bi-}separability, and comparing with
theorem~\ref{thm:ppt-vs-2} we see that there remains only a factor of
at most $2$ to be gained as long as one is restricted to this weaker
constraint.
Is it possible to significantly improve this factor when judging the
performance of SEP or LOCC measurements? In particular, do there
exist constants $C > 0$ and $\alpha < 1$ such that for all $K$ and all
sufficiently large total dimensions $D$ there is a Hermitian
$\Delta \neq 0$ with
\[
  \| \Delta \|_{\mathbf{LOCC}} \leq C \frac{\alpha^K}{\sqrt{D}} \| \Delta \|_1,
  \ \text{ or even }\
  \| \Delta \|_{\mathbf{SEP}} \leq C \frac{\alpha^K}{\sqrt{D}} \| \Delta \|_1 \text{ ?}
\]

\acknowledgments
We thank Fernando Brand\~{a}o, Aram Harrow and
Ashley Montanaro for spurring our interest in generalising the $4$-design
results of~\cite{AmbainisEmerson,MWW} to multiple parties.
Their observations directed our attention towards $2$-norms as the
correct framework for understanding state discrimination with a fixed measurement.

This work was done as part of CL's research placement at the
University of Bristol.
AW is or was supported by the European Commission (STREP ``QCS'' and
Integrated Project ``QESSENCE''), the ERC (Advanced Grant ``IRQUAT''),
a Royal Society Wolfson Merit Award and a Philip Leverhulme Prize.
The Centre for Quantum Technologies is funded by the Singapore
Ministry of Education and the National Research Foundation as part
of the Research Centres of Excellence programme.



\appendix

\section{Proof of theorem~\ref{thm:4-designs-times-k} -- moment inequality}
\label{app:4-design-proofs}
Here we show the missing ingredient, the following moment inequality:
\begin{proposition}
\label{prop:4th-vs-2nd-moment}
For any Hermitian operator $\Delta$ on $\cH = \cH_1\otimes\cdots\otimes\cH_K$,
\begin{equation*}
    \tr\left( \Delta^{\ox 4}
              \left( \sum_{\underline{\pi}\in\mathfrak{S}_4^K} U_{\underline{\pi}} \right)
        \right)              \leq 18^K \left[ \sum_{I\subset[K]} \tr\left(\tr_I \Delta\right)^2 \right]^2.
  \end{equation*}
\end{proposition}

\begin{proof}
In the proof of theorem~\ref{thm:2-designs-times-k} we could easily calculate
\[
  \tr\left( \Delta^{\ox 2}
            \left( \sum_{\underline{\pi}\in\mathfrak{S}_2^K} U_{\underline{\pi}} \right)
     \right)
             = \sum_{I\subset[K]} \tr\left(\tr_I \Delta\right)^2,
\]
because $\mathfrak{S}_2$ only contains 2 elements.
Now, $\mathfrak{S}_4$ contains 24 elements, so to upper bound
\[
  \tr\left( \Delta^{\ox 4}
            \left( \sum_{\underline{\pi}\in\mathfrak{S}_4^K} U_{\underline{\pi}} \right)
     \right)
\]
we will find a way of restricting our attention to only a few of them without loss
of generality. A strategy to do so has already been described when proving the
weaker version of theorem \ref{thm:4-designs-times-k} in section~\ref{sec:2-norms}.

What we have shown there is that, for all $\underline{\sigma}\in\mathfrak{S}_4^K$:
\begin{equation}
\label{split}
  \left| \tr\left(\Delta^{\otimes 4}U_{\underline{\sigma}}\right) \right|
       \leq \frac{1}{2}\tr\left(\Delta^{\otimes 4}U_{\underline{\sigma}^L}\right)
           +\frac{1}{2}\tr\left(\Delta^{\otimes 4}U_{\underline{\sigma}^R}\right),
\end{equation}
with $\underline{\sigma}^L,\underline{\sigma}^R\in\mathfrak{A}^K
:=\{\id,(14),(23),(1234),(1432),(12)(34),(14)(23)\}^K$ given by the splitting map detailed
in the table in Fig.~\ref{fig:split-table}.

Consequently, in order to bound
$\left| \tr \Delta^{\otimes 4} U_{\underline{\sigma}} \right|$ for any
$\underline{\sigma}\in\mathfrak{S}_4^K$, it will be sufficient to bound it for
$\underline{\sigma}\in\mathfrak{A}^K$. Note that for the latter, the trace is
automatically real and non-negative.

Here, we do not use conjugation to reduce even further the number of permutations under consideration to those of $\mathfrak{A}_0:=\{\id,(12)(34),(14)(23)\}$, as was done to prove the weaker version of theorem \ref{thm:4-designs-times-k} in section~\ref{sec:2-norms}. So of course, finding an upper bound to $\left| \tr \Delta^{\otimes 4} U_{\underline{\sigma}} \right|$ for a generic $\underline{\sigma}\in\mathfrak{A}^K$ is a bit more technical than for a generic $\underline{\sigma}\in\mathfrak{A}_0^K$. But what we gain is that we keep track of ``which permutation splits into which pair of permutations'', so that eventually we can make use of an elementary combinatorial argument to get a slightly better factor than the previous $24^K$.

With this aim in view, let us first deal with the following auxiliary problem:\\
Let $\cH=\cA\otimes\cdots\otimes\cG$ be a (finite dimensional)
septempartite Hilbert space.
For a generic operator $X$ on $\cH$ and unit (typically: basis) vectors
$\ket{a},\ket{a'}\in\cA$, \ldots, $\ket{g},\ket{g'}\in\cG$, we denote by
$X_{a,\ldots,g}^{a',\ldots,g'}$ the matrix element
$\bra{a}\cdots\bra{g} X \ket{a'}\cdots\ket{g'}$.

Let $\underline{\sigma}=(\sigma_{\cA},\ldots,\sigma_{\cG})\in\mathfrak{S}_4^7$
be a septuple of permutations.
Now we have, with the $\ket{a_q},\ldots,\ket{g_q}$
($1\leq q\leq 4$) running over an orthonormal basis of $\cA,\ldots,\cG$,
respectively:
\[
  \tr \Delta^{\otimes 4}(U_{\sigma_{\cA}}\otimes\cdots\otimes U_{\sigma_{\cG}})
      = \underset{a_4,\ldots,g_4}{\underset{a_3,\ldots,g_3}{\underset{a_2,\ldots,g_2}{\underset{a_1,\ldots,g_1}{\sum}}}}
                               \prod_{q=1}^4 \Delta_{a_q,\ldots,g_q}^{a_{\sigma_{\cA}(q)},\ldots,g_{\sigma_{\cG}(q)}}.
\]

In the particular case of all the seven permutations in $\mathfrak{A}$,
$\sigma_{\cA}=\id$,
$\sigma_{\cB}=(14)$, $\sigma_{\cC}=(23)$, $\sigma_{\cD}=(1234)$,
$\sigma_{\cE}=(1432)$, $\sigma_{\cF}=(12)(34)$ and $\sigma_{\cG}=(14)(23)$, this becomes
\[\begin{split}
  \tr \Delta^{\otimes 4}(U_{\sigma_{\cA}}\otimes\cdots\otimes U_{\sigma_{\cG}})
         &= \underset{a_4,\ldots,g_4}{\underset{a_3,\ldots,g_3}{\underset{a_2,\ldots,g_2}{\underset{a_1,\ldots,g_1}{\sum}}}}
                  \Delta_{a_1b_1c_1d_1e_1f_1g_1}^{a_1b_4c_1d_2e_4f_2g_4}
                  \Delta_{a_2b_2c_2d_2e_2f_2g_2}^{a_2b_2c_3d_3e_1f_1g_3}
                  \Delta_{a_3b_3c_3d_3e_3f_3g_3}^{a_3b_3c_2d_4e_2f_4g_2}
                  \Delta_{a_4b_4c_4d_4e_4f_4g_4}^{a_4b_1c_4d_1e_3f_3g_1}\\
         &= \underset{b_4,d_4,\ldots,g_4}{\underset{c_3,\ldots,g_3}{\underset{c_2,\ldots,g_2}{\underset{b_1,d_1,\ldots,g_1}{\sum}}}}
                  \left[(\tr_{\cA\otimes\cC}\Delta)^{\Gamma_{\cE}}\right]_{b_1d_1e_4f_1g_1}^{b_4d_2e_1f_2g_4}
                  \left[(\tr_{\cA\otimes\cB}\Delta)^{\Gamma_{\cE}}\right]_{c_2d_2e_1f_2g_2}^{c_3d_3e_2f_1g_3} \\
         &\phantom{= \sum_p^K ====}
            \times \left[(\tr_{\cA\otimes\cB}\Delta)^{\Gamma_{\cE}}\right]_{c_3d_3e_2f_3g_3}^{c_2d_4e_3f_4g_2}
                   \left[(\tr_{\cA\otimes\cC}\Delta)^{\Gamma_{\cE}}\right]_{b_4d_4e_3f_4g_4}^{b_1d_1e_4f_3g_1},
\end{split}\]
where $\Gamma_{\cE}$ denotes the partial transposition on $\cE$.

We can rewrite this using the maximally entangled
$\displaystyle{\Phi_{\cF\ox\cF} = \sum_{ff'} \ketbra{ff}{f'f'}}$:

Letting $\cJ:=\cC\otimes\cD\otimes\cE\otimes\cG$,
$P:=\left(\tr_{\cA\otimes\cB}\Delta\right)^{\Gamma_{\cE}}$ and
$R:=(P\otimes\1_{\cF})(\1_{\cJ}\otimes \Phi_{\cF\otimes\cF})(P\otimes\1_{\cF})$,
we notice that, for all $j,j',f,f',\widetilde{f},\widetilde{f}'$:
\[
  R_{j,f,\widetilde{f}}^{j',f',\widetilde{f}'}
     = \underset{\widetilde{f}'',\widetilde{f}'''}{\underset{f'',f'''}{\underset{j'',j'''}{\sum}}}
           \left(P_{j,f}^{j'',f''}\delta_{\widetilde{f}''=\widetilde{f}}\right)
           \left(\delta_{j'''=j''}\delta_{\widetilde{f}''=f'',\widetilde{f}'''=f'''}\right)
           \left(P_{j''',f'''}^{j',f'}\delta_{\widetilde{f}'''=\widetilde{f}'}\right)
     =\sum_{j''}P_{j,f}^{j'',\widetilde{f}}P_{j'',\widetilde{f}'}^{j',f'}.
\]

Likewise, letting $\cK:=\cB\otimes\cD\otimes\cE\otimes\cG$,
$Q:=(\tr_{\cA\otimes\cC}\Delta)^{\Gamma_{\cE}}$ and
$S:=(Q\otimes\1_{\cF})(\1_{\cK}\otimes \Phi_{\cF\otimes\cF})(Q\otimes\1_{\cF})$,
we have for all $k,k',f,f',\widetilde{f},\widetilde{f}'$:
\[
  S_{k,f',\widetilde{f}'}^{k',f,\widetilde{f}}
       =\sum_{k''}Q_{k,f'}^{k'',\widetilde{f}'}Q_{k'',\widetilde{f}}^{k',f}.
\]

We now just have to make the following identifications:
\[\begin{split}
  \bullet\ & j:=(c_2,d_2,e_1,g_2),\ \ j':=(c_2,d_4,e_3,g_2),\ \ j'':=(c_3,d_3,e_2,g_3), \\
  \bullet\ & k:=(b_4,d_4,e_3,g_4),\ \ k':=(b_4,d_2,e_1,g_4),\ \ k'':=(b_1,d_1,e_4,g_1), \\
  \bullet\ & f:=f_2,\ \ f':=f_4,\ \ \widetilde{f}:=f_1,\ \ \widetilde{f}':=f_3,
\end{split}\]
and to notice that we can actually sum over $j''$ and $k''$ independently.
We thus get:
\[\begin{split}
  \tr \Delta^{\otimes 4}(U_{\sigma_{\cA}}\otimes\cdots\otimes U_{\sigma_{\cG}})
      &= \underset{b_4,d_4,f_4,g_4}{\underset{e_3,f_3}{\underset{c_2,d_2,f_2,g_2}{\underset{e_1,f_1}{\sum}}}}
                R_{c_2,d_2,e_1,g_2,f_2,f_1}^{c_2,d_4,e_3,g_2,f_4,f_3}
                S_{b_4,d_4,e_3,g_4,f_4,f_3}^{b_4,d_2,e_1,g_4,f_2,f_1} \\
      &= \underset{d_4,f_4}{\underset{e_3,f_3}{\underset{d_2,f_2}{\underset{e_1,f_1}{\sum}}}}
               \left(\tr_{\cC\otimes\cG}R\right)_{d_2,e_1,f_2,f_1}^{d_4,e_3,f_4,f_3}
               \left(\tr_{\cB\otimes\cG}S\right)_{d_4,e_3,f_4,f_3}^{d_2,e_1,f_2,f_1}  \\
      &= \tr_{\cD\otimes\cE\otimes\cF\otimes\cF}
               \left(\tr_{\cC\otimes\cG}R\right)\left(\tr_{\cB\otimes\cG}S\right).
\end{split}\]

Defining $\widetilde{P}:=(P\otimes\1_{\cF})(\1_{\cJ}\otimes\sum_f \ket{ff})$ and
$\widetilde{Q}:=(Q\otimes\1_{\cF})(\1_{\cJ}\otimes\sum_f \ket{ff})$,
we see that $R=\widetilde{P}\widetilde{P}^{\dagger}$ and $S=\widetilde{Q}\widetilde{Q}^{\dagger}$.
Hence $R$ and $S$ are positive semidefinite, and so are $\tr_{\cC\otimes\cG}R$
and $\tr_{\cB\otimes\cG}S$.
Thus, using the fact that, for positive semidefinite
$V$ and $W$, $\tr VW \leq (\tr V)(\tr W)$, we obtain
\[
  \tr_{\cD\otimes\cE\otimes\cF\otimes\cF}
      \left[\left(\tr_{\cC\otimes\cG}R\right)\left(\tr_{\cB\otimes\cG}S\right)\right]
  \leq \left(\tr_{\cC\ox\cD\ox\cE\ox\cF\ox\cF\ox\cG}R\right)
       \left(\tr_{\cB\ox\cD\ox\cE\ox\cF\ox\cF\ox\cG}S\right).
\]
On right hand side,
\[ \begin{split}
  \tr R &=\tr_{\cC\ox\cD\ox\cE\ox\cF\ox\cG} P^2 \\
        &=\tr_{\cC\ox\cD\ox\cE\ox\cF\ox\cG} \left((\tr_{\cA\otimes\cB}\Delta)^{\Gamma_{\cE}}\right)^2 \\
        &=\tr_{\cC\ox\cD\ox\cE\ox\cF\ox\cG} \left(\tr_{\cA\otimes\cB}\Delta\right)^2,
\end{split}\]
and likewise,
$\tr S=\tr_{\cB\ox\cD\ox\cE\ox\cF\ox\cG} \left(\tr_{\cA\otimes\cC}\Delta \right)^2$. So, we eventually arrive at
\begin{equation}
  \label{bound}
  \tr \Delta^{\otimes 4}(U_{\sigma_{\cA}}\otimes\cdots\otimes U_{\sigma_{\cG}})
       \leq \left[\tr_{\cC\ox\cD\ox\cE\ox\cF\ox\cG}\left(\tr_{\cA\otimes\cB}\Delta\right)^2\right]
            \left[\tr_{\cB\ox\cD\ox\cE\ox\cF\ox\cG}\left(\tr_{\cA\otimes\cC}\Delta\right)^2\right].
\end{equation}

With this inequality as a tool, we can now return to our initial problem:
For all $\underline{\pi}\in\mathfrak{A}^K=\{\id,(14),(23),(1234),(1432),(12)(34),(14)(23)\}^K$,
we can define the following factors of the global Hilbert space $\cH$:
\begin{align*}
  \cA(\underline{\pi}) &:= \bigotimes_{j\text{ s.t. }\pi_j=\id} \cH_j, \qquad
  \cB(\underline{\pi})  := \bigotimes_{j\text{ s.t. }\pi_j=(14)} \cH_j, \qquad
  \cC(\underline{\pi})  := \bigotimes_{j\text{ s.t. }\pi_j=(23)} \cH_j, \\
  \cD(\underline{\pi}) &:= \bigotimes_{j\text{ s.t. }\pi_j=(1234)} \cH_j, \qquad
  \cE(\underline{\pi})  := \bigotimes_{j\text{ s.t. }\pi_j=(1432)} \cH_j, \\
  \cF(\underline{\pi}) &:= \bigotimes_{j\text{ s.t. }\pi_j=(12)(34)} \cH_j, \qquad
  \cG(\underline{\pi})  := \bigotimes_{j\text{ s.t. }\pi_j=(14)(23)} \cH_j,
\end{align*}
so that clearly,
$\cH=\cA(\underline{\pi})\otimes \cB(\underline{\pi})\otimes \cC(\underline{\pi})
             \otimes \cD(\underline{\pi})\otimes \cE(\underline{\pi})\otimes \cF(\underline{\pi})
             \otimes \cG(\underline{\pi})$.
Hence, using successively the two inequalities (\ref{split}) and
(\ref{bound}), we have:
\[\begin{split}
  \sum_{\underline{\sigma}\in\mathfrak{S}_4^K}
                 \tr \left( \Delta^{\otimes 4}U_{\underline{\sigma}} \right)
     &\leq \sum_{\underline{\sigma}\in\mathfrak{S}_4^K}
      \left\{ \frac{1}{2}\tr\left(\Delta^{\otimes 4}U_{\underline{\sigma}^L}\right)
       +\frac{1}{2}\tr\left(\Delta^{\otimes 4}U_{\underline{\sigma}^R}\right) \right\} \\
     &\leq \sum_{\underline{\sigma}\in\mathfrak{S}_4^K}
       \Biggl\{\frac{1}{2}
        \left[\tr\left(\tr_{\cA(\underline{\sigma}^L)\otimes\cB(\underline{\sigma}^L)}\Delta\right)^2\right]
        \left[\tr\left(\tr_{\cA(\underline{\sigma}^L)\otimes\cC(\underline{\sigma}^L)}\Delta\right)^2\right]
       \Biggr.                                                                                               \\
     &\phantom{=====}
      + \Biggl. \frac{1}{2}
        \left[\tr\left(\tr_{\cA(\underline{\sigma}^R)\otimes\cB(\underline{\sigma}^R)}\Delta\right)^2\right]
        \left[\tr\left(\tr_{\cA(\underline{\sigma}^R)\otimes\cC(\underline{\sigma}^R)}\Delta\right)^2\right]
       \Biggr\}                                                                                              \\
     &= \sum_{\underline{\sigma}\in\mathfrak{S}_4^K}
         \left[\tr \left(\tr_{\cA(\underline{\sigma}^L)\otimes\cB(\underline{\sigma}^L)}\Delta\right)^2\right]
         \left[\tr \left(\tr_{\cA(\underline{\sigma}^L)\otimes\cC(\underline{\sigma}^L)}\Delta\right)^2\right]\\
     &\leq \sum_{\underline{\sigma}\in\mathfrak{S}_4^K}
       \left\{ \frac{1}{2}
       \left[\tr\left(\tr_{\cA(\underline{\sigma}^L)\otimes\cB(\underline{\sigma}^L)}\Delta\right)^2\right]^2
       + \frac{1}{2}
       \left[\tr\left(\tr_{\cA(\underline{\sigma}^L)\otimes\cC(\underline{\sigma}^L)}\Delta\right)^2\right]^2
       \right\}                                                                                              \\
     &= \sum_{\underline{\sigma}\in\mathfrak{S}_4^K}
         \left[\tr\left(\tr_{\cA(\underline{\sigma}^L)\otimes\cB(\underline{\sigma}^L)}\Delta\right)^2\right]^2,
\end{split}\]
where in the last lines we have made use of the symmetry between $\underline{\sigma}^L$
and $\underline{\sigma}^R$ on the one hand, and that between $\cB(\underline{\sigma}^L)$
and $\cC(\underline{\sigma}^L)$ on the other, when $\underline{\sigma}$ ranges over
$\mathfrak{S}_4^K$.

This concludes the main trace estimate. The rest of the argument is combinatorial.
Observe that among the $24$ permutations $\sigma$ of $\mathfrak{S}_4$, two are such that
$\sigma^L=\id$ [namely $\id$ and $(34)$], and four are such that $\sigma^L=(14)$
[namely $(14)$, $(13)$, $(134)$ and $(143)$].
Hence, for all subsets $I \subset [K]$, there are $6^{|I|}\times 18^{K-|I|}$
$K$-tuples of permutations $\underline{\sigma}$ such that $\sigma_i^L$ is either
$\id$ or $(14)$ for $i\in I$, and $\sigma_i^L$ is neither $\id$ nor $(14)$ for
$i\not\in I$. Therefore, we finally obtain:
\[
  \tr \left(\Delta^{\otimes 4}
            \left(\sum_{\underline{\sigma}\in\mathfrak{S}_4^K} U_{\underline{\sigma}}\right)\right)
      \leq 18^K \sum_{I\subset[K]} \left[\tr \left(\tr_I\Delta\right)^2\right]^2
      \leq 18^K \left[ \sum_{I\subset[K]} \tr \left(\tr_I\Delta\right)^2 \right]^2,
\]
which is what we wanted to prove.
\end{proof}

\end{document}